\documentclass[a4paper,10pt]{article}

\usepackage{times}
\usepackage{soul}
\usepackage{url}
\usepackage[hidelinks]{hyperref}
\usepackage[utf8]{inputenc}
\usepackage[small]{caption}
\usepackage{graphicx}
\usepackage{amsmath}
\usepackage{amsthm}
\usepackage{booktabs}
\usepackage{algorithm}
\usepackage[noend]{algpseudocode}
\usepackage{stmaryrd}
\usepackage{amssymb}

\newtheorem{theorem}{Theorem}
\usepackage{newfloat}
\usepackage{listings}
\floatstyle{ruled}
\newfloat{listing}{tb}{lst}{}
\floatname{listing}{Listing}
\usepackage{tikz}
\usepackage{ifthen}
\usetikzlibrary{shapes.geometric, arrows}
\tikzstyle{process} = [rectangle, minimum width=3em, minimum height=2em, text centered, draw=black]
\tikzstyle{arrow} = [thick,->,>=stealth]
\usetikzlibrary{shapes,backgrounds,arrows,calc,positioning,snakes}

\usepackage{xspace}

\tikzstyle{every initial by arrow}=[initial text=] 
\tikzstyle{every state}=[fill=none,draw=black,text=black,inner sep=1pt,minimum size=1mm] 
\tikzstyle{every picture}=[->,>=stealth',shorten >=1pt,auto,node distance=1.3cm, semithick]

\tikzstyle{every node} =
[draw = none, fill = white, thin]
\tikzstyle{every edge} +=
[black, thick]

\tikzstyle{noall} =
[draw = none, fill = none]
\tikzstyle{nodraw} =
[draw = none, fill = white]
\tikzstyle{nofill} =
[draw = black, fill = none]

\tikzstyle{cnode} =
[circle, draw = black]
\tikzstyle{snode} =
[regular polygon, regular polygon sides = 4, draw = black]
\tikzstyle{lnode} =
[diamond, draw = black]

\newcommand{\Ag}{{Ag}}
\newcommand{\wf}[1]{\overline{#1}}
\newcommand{\myi}{{(i)}\xspace}
\newcommand{\myii}{{(ii)}\xspace}

\newcommand{\naww}[1]{{\langle\!\langle #1 \rangle\!\rangle}}
\newcommand{\all}[1]{{[\![ #1 ]\!]}}
\newcommand{\set}[1]{\{#1\}}

\newtheorem{definition}{Definition}
\newtheorem{lemma}{Lemma}
\newtheorem{remark}{Remark}

\newboolean{noednotes}
\setboolean{noednotes}{false}
\newcommand{\ednote}[1]{\ifthenelse{\boolean{noednotes}}{}{
    \marginpar{\colorbox{yellow}{\parbox{\linewidth}{\begin{center}\vspace{-3mm}\footnotesize\textsf{#1}\vspace{-3mm}\end{center}}}}}}
\newcommand{\signedednote}[2]{\ifthenelse{\boolean{noednotes}}{}{
\begin{center} \begin{minipage}{3in}
#2 --    #1
\end{minipage} \end{center}}}

\newcommand{\sem}[1]{\llbracket {#1} \rrbracket}

\title{Towards the Combination of \\ Model Checking and Runtime Verification \\on Multi-Agent Systems}
\author{
	Angelo Ferrando\textsuperscript{\rm 1} and
	Vadim Malvone\textsuperscript{\rm 2}\\
	{\small\textsuperscript{\rm 1}University of Genova, Italy}\\
	{\small\textsuperscript{\rm 2}T\'el\'ecom Paris, France}\\
	{\small angelo.ferrando@unige.it, vadim.malvone@telecom-paris.fr}
}

\date{}

\begin{document}
	
\maketitle




\begin{abstract}
Multi-Agent Systems (MAS) are notoriously complex and hard to verify. In fact, it is not trivial to model a MAS, and even when a model is built, it is not always possible to verify, in a formal way, that it is actually behaving as we expect. Usually, it is relevant to know whether an agent is capable of fulfilling its own goals. One possible way to check this is through Model Checking. Specifically, by verifying Alternating-time Temporal Logic (ATL) properties, where the notion of strategies for achieving goals can be described. Unfortunately, the resulting model checking problem is not decidable in general. 
In this paper, we present a verification procedure based on combining Model Checking and Runtime Verification, where sub-models of the MAS model belonging to decidable fragments are verified by a model checker, and runtime monitors are used to verify the rest. 
We present our technique and we show experimental results.
\end{abstract}


\section{Introduction}

Intelligent systems, such as Multi-Agent Systems (MAS), can be seen as a set of intelligent entities capable of proactively decide how to act to fulfill their own goals.
These entities, called generally agents, are notoriously autonomous, \textit{i.e.}, they do not expect input from an user to act, and social, \textit{i.e.}, they usually communicate amongst each other to achieve common goals.

Software systems are not easy to trust in general. This is especially true in the case of complex and distributed systems, such as MAS. Because of this, we need verification techniques to verify that such systems behave as expected. More specifically, in the case of MAS, it is relevant to know whether the agents are capable of achieving their own goals, by themselves or by collaborating with other agents by forming a coalition. This is usually referred to as the process of finding a strategy for the agent(s).

A well-known formalism for reasoning about strategic behaviours in MAS is Al-ternating-time Temporal Logic ($ATL$)~\cite{AHK02}. 
Before verifying $ATL$ specifications, two questions need to be answered: (i) \emph{does each agent know everything about the system?} (ii) \emph{does the property require the agent to have memory of the system?}
The first question concerns the model of the MAS. If each agent can distinguish each state of the model, then we have \emph{perfect information}; otherwise, we have \emph{imperfect information}.
The second question concerns the $ATL$ property. If the property can be verified without the need for the agent to remember which states of the model have been visited before, then we have \emph{imperfect recall}; otherwise, we have \emph{perfect recall}. 

The model checking problem for $ATL$ giving a generic MAS is known to be undecidable. This is due to the fact that the model checking problem for $ATL$ specifications under imperfect information and perfect recall has been proved to be undecidable~\cite{DimaT11}.
Nonetheless, decidable fragments exist. Indeed, model checking $ATL$ under perfect information is PTIME-complete~\cite{AHK02}, while under imperfect information and imperfect recall is PSPACE~\cite{Sch04}. 
Unfortunately, MAS usually have imperfect information, and when memory is needed to achieve the goals, the resulting model checking problem becomes undecidable.
Given the relevance of the imperfect information setting, even partial solutions to the problem are useful.

This is not the first time that a verification technique alone is not enough to complete the wanted task. 
Specifically, even if the verification of the entire model is not possible, there might still be sub-models of the model for which it is. Consequently, we could focus on these sub-models for which the model checking problem is still decidable;
which are the sub-models with perfect information and perfect recall strategies. 
With more detail, given an $ATL$ formula $\varphi$ and a model of MAS $M$, our procedure extracts all the sub-models of $M$ with perfect information that satisfy a sub-formula of $\varphi$. After this step, runtime monitors are used to check if the remaining part of $\varphi$ can be satisfied at execution time. If this is the case, we can conclude at runtime the satisfaction of $\varphi$ for the corresponding system execution. This is determined by the fact that the system has been observed behaving as expected, since it has verified at design time the sub-formula $\psi$ of $\varphi$, and at runtime the remaining temporal part of $\varphi$ (which consists in the part left to verify in $\varphi$, not covered by $\psi$). Note that, this does not imply that the system satisfies $\varphi$, indeed future executions may violate $\varphi$. The formal result over $\varphi$ only concerns the current system execution, and how it has behaved in it. However, we will present preservation results on the initial model checking problem of $\varphi$ on the model of the system $M$, as well. This will be obtained by linking the result obtained at runtime, with its static counterpart. Hence, we are going to show how the satisfaction (resp., violation) of $\varphi$ at runtime in our approach can be propagated to the verification question over $\varphi$ on model $M$.
Before moving on with the related works in literature, it is important to linger on the main contribution of this work. As we mentioned previously, the problem of statically verify MAS with imperfect information and using perfect recall strategies is undecidable. Thus, the work presented in this paper cannot answer the same question (\textit{i.e.}, we are not claiming decidability for a well-known undecidable problem). Instead, it is focused on gathering and extracting more information about the MAS under analysis at runtime, through runtime verification. This information can be used to better understand the system, and it is an improvement w.r.t. the undecidability of the original problem.

The intuition behind this work lies behind the relation amongst what can be observed at execution time (runtime), and what can be concluded at design time (statically). To the best of our knowledge, no such relation has ever been explored before in the strategic scenario. Usually, static verification of MAS mainly consists in verifying whether strategies for the agents exist to achieve some common goal (expressed as some sort of temporal property enriched with strategic flavour). Even though the two formal verification techniques may seem completely orthogonal, they are very close to each other. In fact, standard runtime verification of temporal properties (such as LTL) consists, in a certain way, in applying model checking at runtime over the all possible executions of a system (whose model may not be available). For the verification of strategic properties as well such relation holds. However, because of the gap between the linearity of the properties verifiable by a runtime monitor, and the branching behaviour of strategic properties, the results that can be obtained through runtime verification are not so natural to propagate to the corresponding model checking problem. Which means, it is not obvious, given a result at runtime, to know what to conclude on the corresponding static verification problem. This is of paramount difference w.r.t. LTL, where a runtime violation can be propagated to a violation of the model checking problem as well. Nonetheless, as we are going to show in this paper, also for strategic properties it is possible to use runtime verification to propagate results on the initial model checking problem. In a nutshell, since it will be better clarified in the due course, static verification of strategic properties over a MAS consists in checking whether a strategy for a set of agents (coalition) can be used to achieve a common (temporal) goal. Now, this is done by analysing, through model checking, the possible executions inside the model in accordance with the strategies for the coalition. Even though at runtime such thorough analysis cannot be done, the observation of an execution of the system at runtime can bring much information. For instance, let us say that the current system execution satisfies the temporal property (the goal, without considering the strategic aspects). Then, this means that the agents at runtime were capable (at least once) to collaborate with each other to achieve a common goal (the temporal property). Note that, this does not imply that the agents will always behave (we are still not exhaustive at runtime), but gives us a vital information about the system: ``if the agents want to achieve the goal, they can''. This runtime outcome can be propagated back to the initial model checking problem, and helps us to conclude the satisfaction of the strategic property when all the agents are assumed to collaborate (one single big coalition). Naturally, it might be possible that even with smaller coalitions the goal would still be achievable, but this is something that cannot be implicated with the only runtime information. On the other hand, if at runtime we observe a wrong behaviour, it means the agents were not capable of achieving the goal. Since we cannot claim which (if any) coalitions were actually formed to achieve the goal, we cannot assume that it is not possible with a greater coalition to achieve the goal. 
In fact, two scenarios are possible. 1) The agents did not form any coalition (each agent works alone). 2) The agents did form a coalition, but this was not enough to achieve the goal. In both cases, there is a common result that can be propagated back to the initial model checking problem, which is that without cooperating the agents cannot achieve the goal. This is true in case (1), since it is what has actually happened at runtime, and it is also true in (2), since by knowing that cooperating (at a certain level) is not enough to achieve the goal, it is also true that with less cooperation the same goal cannot be achieved neither. Note that, this does not imply that the agents will always wrongly behave, indeed with a greater coalition of agents it might still be possible to conclude the goal achievement. The vital information obtained in this way at runtime can be rephrased as: ``if the agents do not cooperate, they cannot achieve the goal''.


\section{Related Work}
\paragraph{Model Checking on MAS.}{
Several approaches for the verification of specifications in $ATL$ and $ATL^*$ under imperfect information and perfect recall have been recently put forward. In one line, restrictions are made on how information is shared amongst the agents, so as to retain decidability~\cite{DBLP:conf/atal/BerthonMM17,DBLP:journals/tocl/BerthonMMRV21}. In a related line, interactions amongst agents are limited to public actions only~\cite{BelardinelliLMR17,DBLP:journals/ai/BelardinelliLMR20}. These approaches are markedly different from ours as they seek to identify classes for which verification is decidable.
Instead, we consider the whole class of iCGS and define a general verification procedure.
In this sense, existing approaches to approximate $ATL$ model checking under imperfect information and perfect recall have either focused on an approximation to perfect information~\cite{BelardinelliLM19,BelardinelliM20} or developed notions of bounded recall~\cite{BelardinelliLM18}. Related to bounded strategies, in \cite{DBLP:journals/ai/JamrogaMM19} the notion of natural strategies is introduced and in \cite{DBLP:conf/atal/JamrogaMM19} is provided a model checking solution for a variant of ATL under imperfect information.

Differently from these works, we introduce, for the first time, a technique that couples model checking and runtime verification to provide results. 
Furthermore, we always concludes with a result. Note that the problem is undecidable in general, thus the result might be inconclusive (but it is always returned). When the result is inconclusive for the whole formula, we present sub-results to give at least the maximum information about the satisfaction/violation of the formula under exam.
}

 
\paragraph{Runtime Verification.}{
Runtime Verification (RV) has never been used before in a strategic context, where monitors check whether a coalition of agents satisfies a strategic property. This can be obtained by combining Model Checking on MAS with RV. 
The combination of Model Checking with RV is not new; in a position paper dating back to 2014, Hinrichs et al. suggested to ``model check what you can, runtime verify the rest'' \cite{DBLP:conf/birthday/HinrichsSZ14}. Their work  presented several realistic examples where such mixed approach would give advantages, but no technical aspects were addressed. Desai et al.~\cite{DBLP:conf/rv/DesaiDS17}  present a framework to combine model checking and runtime verification for robotic applications. 
They represent the discrete model of their system and extract the assumptions deriving from such abstraction. 
Kejstov{\'{a}} et al.~\cite{DBLP:conf/rv/KejstovaRB17} extended an existing software model checker, DIVINE~\cite{barnat2013divine}, with a runtime verification mode. The system under test  consists of a user program in C or C++, along with the environment. Other blended approaches exist, such as a verification-centric software development process for Java making it possible to write, type check, and consistency check behavioural specifications for Java before writing any code~\cite{DBLP:conf/qsic/ZimmermanK09}. Although it integrates a static checker for Java and a runtime assertion checker, it does not properly integrate model checking and RV. 
In all the previously mentioned works, both Model Checking and RV were used to verify temporal properties, such as LTL. Instead, we focus on strategic properties, we show how combining Model Checking of $ATL^*$ properties with RV, and we can give results; even in scenarios where Model Checking alone would not suffice. Because of this, our work is closer in spirit to~\cite{DBLP:conf/birthday/HinrichsSZ14}; in fact, we use RV to support Model Checking in verifying at runtime what the model checker could not at static time.
Finally, in~\cite{10.5555/3463952.3464230}, a demonstration paper presenting the tool deriving by this work may be found. Specifically, in this paper we present the theoretical foundations behind the tool.
}



\noindent

\section{Preliminaries}\label{sec:preliminaries}

In this section we recall some preliminary notions. 
Given a set $U$, $\wf{U}$ denotes its complement. 
We denote the length of a tuple $v$ as $|v|$, and its $i$-th element as $v_i$.  
For $i \leq |v|$, let $v_{\geq i}$ be the suffix $v_{i},\ldots, v_{|v|}$ of $v$ starting at $v_i$ and $v_{\leq i}$ the prefix $v_{1},\ldots, v_{i}$ of $v$. We denote with $v \cdot w$ the concatenation of the tuples $v$ and $w$.

\subsection{Models for Multi-agent systems}
We start by giving a formal model for Multi-agent Systems by means of concurrent game structures with imperfect information \cite{AHK02,Jamroga-Hoek}.
\begin{definition}
	\label{def:cgs}
A \emph{concurrent game structure with imperfect information (iCGS)} is a tuple $M = \langle \Ag, AP, S, s_I, \{Act_i\}_{i \in \Ag}, \{\sim_i\}_{i \in \Ag}, d, \delta, V \rangle$ such that: 
\begin{itemize} 
	\item $\Ag = \set{1,\dots,m}$ is a nonempty finite set of agents (or players).
	\item $AP$ is a nonempty finite set of atomic propositions (atoms).
	\item $S\neq \emptyset$ is a finite set of {\em states}, with {\em initial state}  $s_I \in S$.
	\item For every $i \in \Ag$, $Act_i$ is a nonempty finite set of {\em actions}. Let $Act = \bigcup_{i \in \Ag} Act_i$ be the set of all actions, and $ACT = \prod_{i \in \Ag} Act_i$ the set of all joint actions. 
	\item For every $i \in \Ag$, $\sim_i$ is a relation of {\em indistinguishability} between states. That is, given states $s, s' \in S$, $s \sim_i s'$ iff $s$ and $s'$ are observationally indistinguishable for agent $i$.
	\item The {\em protocol function} $d: \Ag \times S \rightarrow (2^{Act}\setminus \emptyset)$ defines the availability of actions so that for every $i \in Ag$, $s \in S$, \myi $d(i,s) \subseteq Act_i$ and \myii $s \sim_i s'$ implies $d(i,s) = d(i,s')$.
	\item The {\em (deterministic) transition function} $\delta : S \times ACT \to S$ assigns a successor state $s' = \delta(s, \vec{a})$ to each	state $s\in S$, for every joint action $\vec{a} \in ACT$ such that $a_i \in d(i,s)$ for every $i \in \Ag$, that is, $\vec{a}$ is {\em enabled} at $s$.  
	\item $V: S \rightarrow 2^{AP}$ is the {\em labelling function}.
 \end{itemize}
\end{definition}

By Def.~\ref{def:cgs} an iCGS describes the interactions of a group $\Ag$ of agents, starting from the initial state $s_I \in S$, according to the transition function $\delta$. The latter is constrained by the availability of actions to agents, as specified by the protocol function $d$. Furthermore, we assume that every agent $i$ has imperfect information of the exact state of the system; so in any state $s$, $i$ considers epistemically possible all states $s'$ that are $i$-indistinguishable from $s$ \cite{FHMV95}.
When every $\sim_i$ is the identity relation, \textit{i.e.}, $s \sim_i s'$ iff $s = s'$, we obtain a standard CGS with perfect information \cite{AHK02}. 

Given a set $\Gamma \subseteq \Ag$ of agents and a joint action $\vec{a} \in ACT$, let $\vec{a}_{\Gamma}$ and $\vec{a}_{\wf{\Gamma}}$
be two tuples comprising only of actions for the agents in $\Gamma$ and $\wf{\Gamma}$, respectively. 

A history $h \in S^+$ is a finite (non-empty) sequence of states.  
The indistinguishability relations are extended to histories in a synchronous, point-wise way, \textit{i.e.}, histories $h, h' \in S^+$ are {\em
indistinguishable} for agent $i \in \Ag$, or $h \sim_i h'$, iff \myi $|h| = |h'|$ and \myii for all $j \leq |h|$, $h_j \sim_i h'_j$.

%
%
%

\subsection{Syntax}
To reason about the strategic abilities of agents in iCGS with imperfect information, we use Alternating-time Temporal Logic
$ATL^*$~\cite{AHK02}.
\begin{definition}
	\label{def:ATL*}
State ($\varphi$) and path ($\psi$) formulas in $ATL^*$ are defined as
follows, where $q \in AP$ and $\Gamma \subseteq \Ag$:
\begin{eqnarray*}
\varphi & ::= & q \mid \neg \varphi  \mid \varphi \land \varphi \mid \naww{\Gamma} \psi\\
\psi & ::= & \varphi \mid \neg \psi \mid \psi \land \psi \mid X \psi \mid (\psi U \psi)
\end{eqnarray*}
Formulas in $ATL^*$ are all and only the state formulas.
\end{definition}

As customary, a formula $\naww{\Gamma} \Phi$ is read as ``the agents in coalition $\Gamma$ have a strategy to achieve $\Phi$''. The meaning of linear-time operators {\em next} $X$ and {\em until} $U$ is standard \cite{BaierKatoen08}.  Operators $\all{\Gamma}$, {\em release} $R$, {\em finally} $F$, and {\em globally} $G$ can be introduced as usual.
Formulas in the $ATL$ fragment of $ATL^*$ are obtained from
Def.~\ref{def:ATL*} by restricting path formulas $\psi$ as follows,
where $\varphi$ is a state formula and $R$ is the {\em release}
operator:
\begin{eqnarray*}
\psi & ::= & X \varphi \mid (\varphi U \varphi) \mid (\varphi R \varphi)
\end{eqnarray*}

In the rest of the paper, we will also consider the syntax of ATL$^*$ in negative normal form (NNF):

\begin{eqnarray*}
	\varphi & ::= & q \mid \neg q  \mid \varphi \land \varphi \mid \varphi \vee \varphi \mid \naww{\Gamma} \psi \mid \all{\Gamma} \psi\\
	\psi & ::= & \varphi \mid \psi \land \psi \mid \psi \vee \psi \mid X \psi \mid (\psi U \psi) \mid (\psi R \psi)
\end{eqnarray*}

where $q \in AP$ and $\Gamma \subseteq \Ag$.


\subsection{Semantics}
When giving a semantics to $ATL^*$ formulas we assume that agents are endowed with {\em uniform strategies} \cite{Jamroga-Hoek}, \textit{i.e.}, they perform the same action whenever they have the same information. 

\begin{definition}
	\label{uniformitybouned}
	A \emph{uniform strategy} for agent $i \in \Ag$ is a function $\sigma_i: S^+ \to Act_i$ such that for all histories $h, h' \in S^+$, \myi
	$\sigma_i(h) \in d(i, last(h))$; and \myii
	$h \sim_i h'$ implies $\sigma_i(h) = \sigma_i(h')$.  
\end{definition}

By Def.~\ref{uniformitybouned} any strategy for agent $i$ has to return actions that are enabled for $i$. Also, whenever two histories are indistinguishable for $i$, then the same action is returned. Notice that, for the case of CGS (perfect information), condition \myii is satisfied by any strategy $\sigma$. Furthermore, we obtain memoryless (or imperfect recall) strategies by considering the domain of $\sigma_i$ in $S$, \textit{i.e.}, $\sigma_i: S \to Act_i$.

Given an iCGS $M$, a {\em path} $p \in S^{\omega}$ is an infinite sequence $s_1 s_2\dots$ of states.  
Given a joint strategy $\sigma_\Gamma = \set{\sigma_i \mid i \in \Gamma}$, comprising of one strategy for each agent in coalition $\Gamma$, a path $p$ is \emph{$\sigma_\Gamma$-compatible} iff for every $j \geq 1$, $p_{j+1} = \delta(p_j, \vec{a})$ for some joint action $\vec{a}$ such that for every $i \in \Gamma$, $a_i = \sigma_i(p_{\leq j})$, and for every $i \in \wf{\Gamma}$, $a_i \in d(i,p_j)$.
Let $out(s, \sigma_\Gamma)$ be the set of all $\sigma_\Gamma$-compatible paths from $s$.

We can now assign a meaning to $ATL^*$ formulas on iCGS. 
\begin{definition}
	\label{satisfaction}
The satisfaction relation $\models$ for an iCGS $M$, state $s \in S$, path $p \in S^{\omega}$, atom $q \in AP$, and $ATL^*$ formula $\phi$ is defined as follows: 

\begin{tabbing}
$(M, s) \models q$ \ \ \ \ \ \ \ \ \ \ \=  iff \ \ \ \ \= $q \in {V}(s)$\\
$(M, s) \models \neg \varphi$ \> iff \>  $(M, s) \not \models \varphi$\\
$(M, s) \models \varphi \land \varphi'$ \>iff\>  $(M, s) \models \varphi$  and  $(M, s) \models \varphi'$\\
$(M, s) \models \naww{{\Gamma}} \psi$  \>iff\> for some $\sigma_\Gamma$, for all $p \!\in\! out(s, \sigma_\Gamma)$,  $(M, p) \!\models\! \psi$\\
$(M, p) \models \varphi $ \>iff\> $(M, p_1) \models \varphi$ \\
$(M, p) \models \neg \psi$ \> iff \>  $(M, p) \not \models \psi$\\
$(M, p) \models \psi \land \psi'$ \>iff\>  $(M, p) \models \psi$  and  $(M, p) \models \psi'$\\
$(M, p) \models X \psi$ \>iff\> $(M, p_{\geq 2}) \models \psi$\\
$(M, p) \models \psi U \psi'$  \>iff\> for some $k \geq 1$, $(M, p_{\geq k}) \models \psi'$, and \\ \> \> for all $j$, $1 \leq j < k \Rightarrow (M, p_{\geq j}) \models \psi$
\end{tabbing}
\end{definition}

We say that formula $\phi$ is {\em true} in an iCGS $M$, or $M \models \phi$, iff $(M, s_I) \models \phi$.

We now state the model checking problem.
\begin{definition}
	Given an iCGS $M$ and a formula $\phi$, the model checking problem concerns determining whether $M \models \phi$.
\end{definition}

Since the semantics provided in Def.~\ref{satisfaction} is the standard interpretation of $ATL^*$ \cite{AHK02,Jamroga-Hoek}, it
is well known that model checking $ATL$, {\em a fortiori} $ATL^*$, against iCGS with imperfect information and  perfect recall is undecidable~\cite{DimaT11}.  In the rest of the paper we develop methods to obtain partial solutions to this by using Runtime Verification (RV).

\subsection{Runtime verification and Monitors}

Given a nonempty set of atomic propositions $AP$, we define a \emph{trace} $\rho=ev_1 ev_2 \ldots$, as a sequence of set of events in $AP$ (\textit{i.e.}, for each $i$ we have that $ev_i \in 2^{AP}$). For brevity, we name $\Sigma=2^{AP}$ the powerset of atomic propositions. As usual, $\Sigma^*$ is the \emph{set of all possible finite traces} over $\Sigma$, and $\Sigma^\omega$ is the \emph{set of all possible infinite traces} over $\Sigma$.

The standard formalism to specify formal properties in RV is Linear Temporal Logic (LTL)~\cite{DBLP:conf/focs/Pnueli77}. 
The syntax of LTL is as follows:
\begin{eqnarray*}
\psi & ::= & q \mid \neg \psi \mid \psi \land \psi \mid X \psi \mid (\psi U \psi)
\end{eqnarray*}
where $q\in AP$ is an event (a proposition), $\psi$ is a formula, $U$ stands for \emph{until}, and $X$ stands for \emph{next-time}. 

Let $\rho\in\Sigma^\omega$ be an infinite sequence of events over $\Sigma$, the semantics of LTL is as follows:

\begin{tabbing}
	$\rho \models q$ \ \ \ \ \ \ \ \ \ \ \=  iff \ \ \ \ \= $q \in \rho_1$\\
	$\rho \models \neg \psi$ \> iff \>  $\rho \not\models \psi$\\
	$\rho \models \psi \land \psi'$ \>iff\>  $\rho \models \psi$  and  $\rho \models \psi'$\\
	$\rho \models X \psi$ \>iff\> $\rho_{\geq 2} \models \psi$\\
	$\rho \models \psi U \psi'$  \>iff\> for some $k \geq 1$, $\rho_{\geq k} \models \psi'$, and  for all $j$, $1 \leq j < k \Rightarrow \rho_{\geq j} \models \psi$
\end{tabbing}


Thus, given an LTL property $\psi$, we denote $\sem{\psi}$ the language of the property, \textit{i.e.}, the set of traces which satisfy $\psi$; namely $\sem{\psi} = \{ \rho \;|\; \rho \models \psi\}$.

\begin{definition}[Monitor]\label{rv-def}
Let $AP$ be the alphabet of atomic propositions, $\Sigma=2^{AP}$ be its powerset, and $\psi$ be an LTL property. Then, a monitor for $\psi$ is a function
$Mon_{\psi}:\Sigma^*\rightarrow\mathbb{B}_3$, where
$\mathbb{B}_3=\{\top, \bot, ?\}$:
$$
Mon_{\psi}(\rho) =
\left\{
\bgroup
\def\arraystretch{1.2}
  \begin{tabular}{cl}
  $\top$ & {\qquad$\forall_{\rho' \in \Sigma^\omega} \:\: \rho \cdot \rho' \in \sem{\psi}$}\\
  $\bot$ & {\qquad$\forall_{\rho' \in \Sigma^\omega} \:\: \rho \cdot \rho' \notin\sem{\psi}$}\\
  $?$ & {\qquad$otherwise.$}\\
  \end{tabular}
\egroup
\right.
$$
\end{definition}
Intuitively, a monitor returns~$\top$ if all continuations ($\rho'$) of $\rho$ satisfy $\psi$; $\bot$ if all possible
continuations of $\rho$ violate $\psi$; $?$ otherwise. The first two outcomes are standard representations of satisfaction and violation, while the third is specific to RV. In more detail, it denotes when the monitor cannot conclude any verdict yet. This is closely related to the fact that RV is applied while the system is still running, and not all information about it are available. For instance, a property might be currently satisfied (resp., violated) by the system, but violated (resp., satisfied) in the (still unknown) future. The monitor can only safely conclude any of the two final verdicts ($\top$ or $\bot$) if it is sure such verdict will never change. The addition of the third outcome symbol $?$ helps the monitor to represent its position of uncertainty w.r.t. the current system execution.


\subsection{Negative and Positive Sub-models}

Now, we recall two definitions of sub-models, defined in \cite{DBLP:journals/corr/abs-2112-13621}, that we will use in our verification procedure. We start with the definition of negative sub-models.

\begin{definition}[Negative sub-model]\label{subneg}
	Given an iCGS $M = \langle \Ag, AP, S, s_I,  \{Act_i\}_{i \in \Ag}, \allowbreak \{\sim_i\}_{i \in \Ag}, d, \delta, V \rangle$, we denote with $M_n = \langle \Ag, AP, S_n, s_I, \{Act_i\}_{i \in \Ag}, \{\sim^n_i\}_{i \in \Ag}, d_n, \allowbreak \delta_n,  V_n \rangle$ a negative sub-model of $M$, formally $M_n \subseteq M$, such that:
	\begin{itemize}
		\item the set of states is defined as $S_n = S^\star \cup \{s_\bot\}$, where $S^\star \subseteq S$, and $s_I \in S^\star$ is the initial state.
		\item $\sim^n_i$ is defined as the corresponding $\sim_i$ restricted to $S^\star$.
		\item The protocol function is defined as $d_n : \Ag \times S_n \rightarrow (2^{Act}\setminus\emptyset)$, where $d_n(i,s) = d(i,s)$, for every $s \in S^\star$ and $d_n(i,s_\bot) = Act_i$, for all $i \in \Ag$.
		\item The transition function is defined as $\delta_n : S_n \times ACT \rightarrow S_n$, where given a transition $\delta(s, \vec{a}) = s'$, if $s, s' \in S^\star$ then $\delta_n(s, \vec{a}) = \delta(s, \vec{a}) = s'$ else if $s' \in S \setminus S^\star$ and $s \in S_n$ then $\delta_n(s, \vec{a}) = s_\bot$. 
		\item for all $s\in S^\star$, $V_n(s) = V(s)$ and $V_n(s_\bot) = \emptyset$.
	\end{itemize}
\end{definition}

Now, we present the definition of positive sub-models. 

\begin{definition}[Positive sub-model]\label{subpos}
	Given an iCGS $M = \langle \Ag, AP, S, s_I,  \{Act_i\}_{i \in \Ag}, \allowbreak \{\sim_i\}_{i \in \Ag}, d, \delta, V \rangle$, we denote with $M_p = \langle \Ag, AP, S_p, s_I, \{Act_i\}_{i \in \Ag}, \{\sim^p_i\}_{i \in \Ag}, d_p, \allowbreak \delta_p, V_p \rangle$ a positive sub-model of $M$, formally $M_p \subseteq M$, such that:
	\begin{itemize}
		\item the set of states is defined as $S_p = S^\star \cup \{s_\top\}$, where $S^\star \subseteq S$, and $s_I \in S^\star$ is the initial state.
		\item $\sim^p_i$ is defined as the corresponding $\sim_i$ restricted to $S^\star$.
		\item The protocol function is defined as $d_p : \Ag \times S_p \rightarrow (2^{Act}\setminus\emptyset)$, where $d_p(i,s) = d(i,s)$, for every $s \in S^\star$ and $d_p(i,s_\top) = Act_i$, for all $i \in \Ag$.
		\item The transition function is defined as $\delta_p : S_p \times ACT \rightarrow S_p$, where given a transition $\delta(s, \vec{a}) = s'$, if $s, s' \in S^\star$ then $\delta_p(s, \vec{a}) = \delta(s, \vec{a}) = s'$ else if $s' \in S \setminus S^\star$ and $s \in S_p$ then $\delta_p(s, \vec{a}) = s_\top$. 
		\item for all $s\in S^\star$, $V_p(s) = V(s)$ and $V_p(s_\top) = AP$.
	\end{itemize}
\end{definition}

Note that, the above sub-models are still iCGSs.

We conclude this part by recalling two preservation results presented in \cite{DBLP:journals/corr/abs-2112-13621}.

We start with a preservation result from negative sub-models to the original model.

\begin{lemma}\label{lemma:negsub}
	Given a model $M$, a negative sub-model with perfect information $M_n$ of $M$, and a formula $\varphi$ of the form $\varphi = \naww{A} \psi$ (resp., $\all{A} \psi$) for some $A \subseteq Ag$. For any $s \in S_n \setminus \{s_\bot\}$, we have that:
	\begin{center}
		$M_n, s \models \varphi \Rightarrow M, s \models \varphi$
	\end{center}
\end{lemma}


We also consider the preservation result from positive sub-models to the original model.

\begin{lemma}\label{lemma:possub}
	Given a model $M$, a positive sub-model with perfect information $M_p$ of $M$, and a formula $\varphi$ of the form $\varphi = \naww{A} \psi$ (resp., $\all{A} \psi$) for some $A \subseteq Ag$. For any $s \in S_p \setminus \{s_\top\}$, we have that:
	\begin{center}
		$M_p, s \not\models \varphi \Rightarrow M, s \not \models \varphi$ 
	\end{center} 
\end{lemma}



\section{Our procedure}\label{sec:procedure}
In this section, we provide a procedure to handle games with imperfect information and perfect recall strategies, a problem in general undecidable. 
The overall model checking procedure is described in Algorithm \ref{alg:full}.
It takes in input a model $M$, a formula $\varphi$, and a trace $h$ (denoting an execution of the system) and calls the function $Preprocessing()$ to generate the negative normal form of $\varphi$ and to replace all negated atoms with new positive atoms inside $M$ and $\varphi$. After that, it calls the function $FindSub$-$models()$ to generate all the positive and negative sub-models that represent all the possible sub-models with perfect information of $M$. Then, there is a while loop (lines 4-7) that for each candidate checks the sub-formulas true on the sub-models via $CheckSub$-$formulas()$ and returns a result via $RuntimeVerification()$.
For the algorithms and additional details regarding the procedures $Preprocessing()$, $FindSub$-$models()$, and $CheckSub$-$formulas()$ see \cite{DBLP:journals/corr/abs-2112-13621}.

\begin{algorithm}
	\footnotesize
	\caption{$ModelCheckingProcedure$ ($M$, $\varphi$, $h$)}
	\label{alg:full}
	\begin{algorithmic}[1] 
		\State $Preprocessing(M, \varphi)$;
		\State $candidates$ = $FindSub$-$models$$(M, \varphi)$;
		\State $finalresult = \emptyset$;
		\While {$candidates$ is not empty}
		\State extract $\langle M_n, M_p \rangle$ from $candidates$;
		\State $result$ = $CheckSub$-$\!formulas$$(\langle M_n, M_p \rangle,\varphi)$;
		\State  $finalresult$ = $RuntimeVerification(M, \varphi, h, result)$ $\cup$ $finalresult$; 
		\EndWhile
		\State \Return $finalresult$;
	\end{algorithmic}
\end{algorithm}

Now, we will focus on the last step, the procedure $RuntimeVerification()$. 
It is performed at runtime, directly on the actual system. In previous steps, the sub-models satisfying (resp., violating) sub-properties $\varphi'$ of $\varphi$ are generated, and listed into the set $result$.
In Algorithm~\ref{alg:rv}, we report the algorithm performing runtime verification on the actual system. Such algorithm gets in input the model $M$, an ATL property $\varphi$ to verify, an execution trace $h$ of events observed by executing the actual system, and the set $result$ containing the sub-properties of $\varphi$ that have been checked on sub-models of $M$.
First, in lines 1-4, the algorithm updates the model $M$ with the atoms corresponding to the sub-properties verified previously on sub-models of $M$. This step is necessary to keep track explicitly inside $M$ of where the sub-properties are verified (resp., violated). This last aspect depends on which sub-model had been used to verify the sub-property (whether negative or positive). After that, the formula $\varphi$ needs to be updated accordingly to the newly introduced atoms. This is obtained through updating the formula, by generating at the same time two new versions $\psi_n$ and $\psi_p$ for the corresponding negative and positive versions (lines 6-14). Once $\psi_n$ and $\psi_p$ have been generated, they need to be converted into their corresponding LTL representation to be verified at runtime. Note that, $\psi_n$ and $\psi_p$ are still ATL properties, which may contain strategic operators. Thus, this translation is obtained by removing the strategic operators, leaving only the temporal ones (and the atoms). The resulting two new LTL properties $\varphi_n$ and $\varphi_p$ are so obtained (lines 15-16). Finally, by having these two LTL properties, the algorithm proceeds generating (using the standard LTL monitor generation algorithm~\cite{DBLP:journals/tosem/BauerLS11}) the corresponding monitors $Mon_{\varphi_n}$ and $Mon_{\varphi_p}$. Such monitors are then used by Algorithm~\ref{alg:rv} to check $\varphi_n$ and $\varphi_p$ over an execution trace $h$ given in input. The latter consists in a trace observed by executing the system modelled by $M$ (so, the actual system). Analysing $h$ the monitor can conclude the satisfaction (resp., violation) of the LTL property under analysis. However, only certain results can actually be considered valid. Specifically, when $Mon_{\varphi_n}(h) = \top$, or when $Mon_{\varphi_p}(h) = \bot$. The other cases are considered undefined, since nothing can be concluded at runtime. 
The reason why line 17 and line 20's conditions are enough to conclude $\top$ and $\bot$ (resp.) directly follow from the following lemmas.

\begin{algorithm}
	\footnotesize
	\caption{$RuntimeVerification$ ($M$, $\varphi$, $h$, $result$)}
	\label{alg:rv}
	\begin{algorithmic}[1] 
		\State $k$ = $?$;
		\For{$s \in S$}
		\State take set $atoms$ from $result(s)$;
		\State $UpdateModel$($M$, $s$, $atoms$);	
		\EndFor
		\State $\varphi_{mc} = \emptyset$;
		\For{$\langle s,\psi,atom \rangle \in result$}
		\State $\varphi_{mc} = \varphi_{mc} \cup \psi$;
		\EndFor
		\State $\varphi_{rv} = SubFormulas(\varphi) \setminus \varphi_{mc}$;
		\State $\psi_n = \varphi$, $\psi_p = \varphi$; 
		\While{$result$ is not empty}
		\State extract $\langle s, \psi, vatom_{\psi}\rangle$ from $result$; 
		\If{$v = n$}
		\State $\psi_n$ = $UpdateFormula$($\psi_n$, $\psi$, $natom_{\psi}$);
		\Else
		\State $\psi_p$ = $UpdateFormula$($\psi_p$, $\psi$, $patom_{\psi}$);
		\EndIf
		\EndWhile
		\State $\varphi_n$ = $FromATLtoLTL$($\psi_n$, $n$);
		\State $\varphi_p$ = $FromATLtoLTL$($\psi_p$, $p$);
		\State $Mon_{\varphi_p} = GenerateMonitor(\varphi_p)$;
		\State $Mon_{\varphi_n} = GenerateMonitor(\varphi_n)$;
		\If {$Mon_{\varphi_n}(h) = \top$}
		\State $k$ = $\top$;
		\EndIf
		\If {$Mon_{\varphi_p}(h) = \bot$}
		\State $k$ = $\bot$;
		\EndIf
		\State $\varphi_{unchk} = \emptyset$;
		\For{$\varphi' \in \varphi_{rv}$}
		\State $Mon_{\varphi'} = GenerateMonitor(\varphi')$;
		\If{$Mon_{\varphi'}(h) = \;?$}
		\State $\varphi_{rv} = \varphi_{rv} \setminus \varphi'$;
		\State $\varphi_{unchk} = \varphi_{unchk} \cup \varphi'$;
		\EndIf
		\EndFor
		\State \Return $\langle k, \varphi_{mc}, \varphi_{rv}, \varphi_{unchk} \rangle$;
	\end{algorithmic}
\end{algorithm} 

We start with a preservation result from the truth of the monitor output to ATL$^*$ model checking.

\begin{lemma}\label{lemma:rv1}
	Given a model $M$ and a formula $\varphi$, for any history $h$ of $M$ starting in $s_I$, we have that:
	\begin{center}
		$Mon_{\varphi_{LTL}}(h) = \top \;\implies\; M,s_I \models \varphi_{Ag}$ 
	\end{center}
	where $\varphi_{LTL}$ is the variant of $\varphi$ where all strategic operators are removed and $\varphi_{Ag}$ is the variant of $\varphi$ where all strategic operators are converted into $\naww{Ag}$.
\end{lemma}
\begin{proof}
	First, consider the formula $\varphi = \naww{\Gamma} \psi$, in which $\Gamma \subseteq Ag$ and $\psi$ is a temporal formula without quantifications. So, $\varphi_{LTL} = \psi$ and $\varphi_{Ag} = \naww{Ag} \psi$.
	By Def.\ref{rv-def} we know that $Mon_{\varphi_{LTL}}(h) = \top$ if and only if for all path $p$ in $S^\omega$ we have that $h \cdot p$ is in $\llbracket\varphi_{LTL}\rrbracket$. Note that, the latter is the set of paths that satisfy $\psi$, \textit{i.e.}, $\llbracket\varphi_{LTL}\rrbracket = \{ p \mid M, p \models \psi\}$. By Def.\ref{def:ATL*} we know that $M,s_I \models \varphi_{Ag}$ if and only if there exist a strategy profile $\sigma_{Ag}$ such that for all paths $p$ in $out(s_I,\sigma_{Ag})$ we have that $M,p \models \psi$. Notice that, since the strategic operator involves the whole set of agents, $out(s_I,\sigma_{Ag})$ is composed by a single path.
	Thus, to guarantee that $\varphi_{Ag}$ holds in $M$, our objective is to construct from $s_I$ the history $h$ as prefix of the unique path in $out(s_I,\sigma_{Ag})$. Since we have $\naww{Ag}$ as strategic operator, this means that there is a way for the set of agents to construct $h$ starting from $s_I$ and the set $out(s_I,\sigma_{Ag})$ becomes equal to $\{p\}$, where $p = h \cdot p'$, for any $p' \in S^{\omega}$. From the above reasoning, the result follows.
	
	To conclude the proof, note that if we have a formula with more strategic operators then we can use a classic bottom-up approach.
\end{proof}

Now, we present a preservation result from the falsity of the monitor output to ATL$^*$ model checking.

\begin{lemma}\label{lemma:rv2}
	Given a model $M$ and a formula $\varphi$, for any history $h$ of $M$ starting in $s_I$, we have that:
	\begin{center}
		$Mon_{\varphi_{LTL}}(h) = \bot \;\implies\; M,s_I \not\models \varphi_{\emptyset}$ 
	\end{center}
	where $\varphi_{LTL}$ is the variant of $\varphi$ where all strategic operators are removed and $\varphi_{\emptyset}$ is the variant of $\varphi$ where all strategic operators are converted into $\naww{\emptyset}$.
\end{lemma}

\begin{proof}
	First, consider the formula $\varphi = \naww{\Gamma} \psi$, in which $\Gamma \subseteq Ag$ and $\psi$ is a temporal formula without quantifications. So, $\varphi_{LTL} = \psi$ and $\varphi_{\emptyset} = \naww{\emptyset} \psi$.
	By Def.\ref{rv-def} we know that $Mon_{\varphi_{LTL}}(h) = \bot$ if and only if for all path $p$ in $S^\omega$ we have that $h \cdot p$ is not in $\llbracket\varphi_{LTL}\rrbracket$. Note that, the latter is the set of paths that satisfy $\psi$, \textit{i.e.}, $\llbracket \varphi_{LTL}\rrbracket = \{ p \mid M, p \models \psi\}$. By Def.\ref{def:ATL*} we know that $M,s_I \not \models \varphi_{\emptyset}$ if and only if for all strategy profiles $\sigma_{\emptyset}$, there exists a path $p$ in $out(s_I,\sigma_{\emptyset})$ such that $M,p \not \models \psi$. Notice that, since the strategic operator is empty then $out(s_I,\sigma_{\emptyset})$ is composed by all the paths in $M$.
	Thus, to guarantee that $\varphi_{\emptyset}$ does not hold in $M$, our objective is to select a path $p$ in $out(s_I,\sigma_{\emptyset})$ starting from $s_I$, where $p = h \cdot p'$, for any $p' \in S^{\omega}$. Given the assumption that $h \cdot p$ is not in $\llbracket\varphi_{LTL}\rrbracket$ then the result follows.
	
	To conclude the proof, note that if we have a formula with more strategic operators then we can use a classic bottom-up approach.
\end{proof}

It is important to evaluate in depth the meaning of the two lemmas presented above, we do this in the following remark.

\begin{remark}
	Lemma~\ref{lemma:rv1} and~\ref{lemma:rv2} show a preservation result from runtime verification to ATL$^*$ model checking that needs to be discussed. If our monitor returns true we have two possibilities: 
	\begin{enumerate}
		\item the procedure found a negative sub-model in which the original formula $\varphi$ is satisfied then it can conclude the verification procedure by using RV only by checking that the atom representing $\varphi$ holds in the initial state of the history $h$ given in input;
		\item a sub-formula $\varphi'$ is satisfied in a negative sub-model and at runtime the formula $\varphi_{Ag}$ holds on the history $h$ given in input.
	\end{enumerate}
	While case 1. gives a preservation result for the formula $\varphi$ given in input, case 2. checks formula $\varphi_{Ag}$ instead of $\varphi$. That is, it substitutes $Ag$ as coalition for all the strategic operators of $\varphi$ but the ones in $\varphi'$. So, our procedure approximates the truth value by considering the case in which all the agents in the game collaborate to achieve the objectives not satisfied in the model checking phase. That is, while in~\cite{BelardinelliLM19, BelardinelliM20} the approximation is given in terms of information, in~\cite{BelardinelliLM18} is given in terms of recall of the strategies, and in~\cite{DBLP:journals/corr/abs-2112-13621} the approximation is given by generalizing the logic, here we give results by approximating the coalitions. Furthermore, we recall that our procedure produces always results, even partial. This aspect is strongly relevant in concrete scenario in which there is the necessity to have some sort of verification results. For example, in the context of swarm robots~\cite{DBLP:journals/ai/KouvarosL16}, with our procedure we can verify macro properties such as "the system works properly" since we are able to guarantee fully collaboration between agents because this property is relevant and desirable for each agent in the game. The same reasoning described above, can be applied in a complementary way for the case of positive sub-models and the falsity.
\end{remark}
%


To conclude this section we show and prove the complexity of our procedure.

\begin{theorem}\label{theo:full}
	Algorithm~\ref{alg:full} terminates in $2EXPTIME$. Moreover, Algorithm~\ref{alg:full} is sound: if the value returned is different from $?$, then $M \models \varphi_{Ag}$ iff $k = \top$.
\end{theorem}

\begin{proof}
	The preprocessing phase is polynomial in the size of the model and the formula. As described in \cite{DBLP:journals/corr/abs-2112-13621}, $FindSub$-$models()$ terminates in $EXPTIME$.
	The while loop in lines 3-7 needs to check all the candidates and in the worst case the size of the list of candidates is equal to the size of the set of states of $M$ (\textit{i.e.}, polynomial in the size of $M$). 
	About $CheckSub$-$formulas()$, as described in \cite{DBLP:journals/corr/abs-2112-13621}, the complexity is $2EXPTIME$ due to the ATL$^*$ model checking that is called in it.
	Finally, Algorithm \ref{alg:rv} terminates in $2EXPTIME$. In particular, loops in lines 2, 6, and 10 terminate in polynomial time with respect to the size of the model and the size of the formula. As described in \cite{DBLP:journals/tosem/BauerLS11}, to generate a monitor requires $2EXPTIME$ in the size of the formula and the execution of a monitor is linear in the size of the formula. 
	So, the total complexity is still determined by the subroutines and directly follows. 
	
	About the soundness, suppose that the value returned is different from $?$. In particular, either $k = \top$ or $k = \bot$. 
	If $M \models \varphi_{Ag}$ and $k = \bot$, then by Algorithm~\ref{alg:full} and~\ref{alg:rv}, we have that $Mon_{\varphi_p} (h) = \bot$.
	Now, there are two cases: (1) $h$ is an history of $M$ (2) there exists an history $h'$ of $M$ that differs from $h$ for some atomic propositions added to $h$ in lines 2-4 of Algorithm~\ref{alg:rv}.
	For (1), we know that $h$ is in $M$ and thus $Mon_{\varphi_p} (h) = \bot$ implies $M \not\models \varphi_{\emptyset}$ by Lemma~\ref{lemma:rv2} that implies $M \not\models \varphi_{Ag}$ by the semantics in Def.~\ref{satisfaction}, a contradiction. Hence, $k = \top$ as required. 
	For (2), suppose that $h$ has only one additional atomic proposition $atom_\psi$. 
	The latter means that $CheckSub-formulas()$ found a positive sub-model $M_p$ in which $M_p,s \models \psi$, for some $s \in S_p$. 
	By Lemma~\ref{lemma:possub}, for all $s \in S_p$, we know that if $M_p,s \not \models \psi$ then $M,s \not\models \psi$. So, $h$ over-approximates $h'$, i.e. there could be some states that in $h$ are labeled with $atom_\psi$ but they don't satisfy $\psi$ in $h$. Thus, if $Mon_{\varphi_p} (h) = \bot$ then $M \not\models \varphi_{\emptyset}$ by Lemma~\ref{lemma:rv2} that implies $M \not\models \varphi_{Ag}$, a contradiction. Hence, $k = \top$ as required. 
	Obviously, we can generalize the above reasoning in case $h$ and $h'$ differ for multiple atomic propositions.
	On the other hand, if $k = \top$ then by Algorithm~\ref{alg:full} and~\ref{alg:rv}, we have that $Mon_{\varphi_n} (h) = \top$.
	Again, there are two cases: (1) $h$ is an history of $M$ (2) there exists an history $h'$ of $M$ that differs from $h$ for some atomic propositions added to $h$ in lines 2-4 of Algorithm~\ref{alg:rv}.
	For (1), we know that $h$ is in $M$ and thus $Mon_{\varphi_n} (h) = \top$ implies $M \models \varphi_{Ag}$ by Lemma~\ref{lemma:rv1} as required. 
	For (2), suppose that $h$ has only one additional atomic proposition $atom_\psi$. The latter means that $CheckSub-formulas()$ found a negative sub-model $M_n$ in which $M_n,s \models \psi$, for some $s \in S_n$. By Lemma~\ref{lemma:negsub}, for all $s \in S_n$, we know that if $M_n,s \models \psi$ then $M,s \models \psi$. So, $h$ under-approximates $h'$, i.e. there could be some states that in $h$ are not labeled with $atom_\psi$ but they satisfy $\psi$ in $M$. Thus, if $Mon_{\varphi_n} (h) = \top$ then $M \models \varphi_{Ag}$ by Lemma~\ref{lemma:rv1}, as required.
\end{proof}

\section{Our tool}\label{sec: tool}

The algorithms presented previously have been implemented in Java\footnote{The tool can be found at \url{https://github.com/AngeloFerrando/StrategyRV}}. 
The resulting tool implementing Algorithm~\ref{alg:full} allows to extract all sub-models with perfect information ($CheckSub$-$formulas()$) that satisfy a strategic objective from a model given in input. The extracted sub-models, along with the corresponding sub-formulas, are then used by the tool to generate and execute the corresponding monitors over a system execution (Algorithm~\ref{alg:rv}).

In more detail, as shown in Figure~\ref{fig:tool}, the tool expects a model in input formatted as a Json file. This file is then parsed, and an internal representation of the model is generated. After that, the verification of a sub-model against a sub-formula is achieved by translating the sub-model into its equivalent ISPL (Interpreted Systems Programming Language) program, which then is verified by using the model checker MCMAS\footnote{\url{https://vas.doc.ic.ac.uk/software/mcmas/}}\cite{LomuscioRaimondi06c}. This corresponds to the verification steps performed in $CheckSub$-$formulas()$ (\textit{i.e.}, where static verification through MCMAS is used).
For each sub-model that satisfies this verification step, the tool produces a corresponding tuple; which contains the information needed by Algorithm~\ref{alg:rv} to complete the verification at runtime. 

\tikzset{every picture/.style={line width=0.75pt}} 

\begin{figure}[ht]
\centering

\scalebox{0.50}{

\begin{tikzpicture}[x=0.75pt,y=0.75pt,yscale=-1,xscale=1]

\draw   (216,109) -- (322,109) -- (322,168.8) -- (216,168.8) -- cycle ;
\draw    (105,85) .. controls (144.2,55.6) and (172.84,147.98) .. (211.61,123.45) ;
\draw [shift={(214,121.8)}, rotate = 143.13] [fill={rgb, 255:red, 0; green, 0; blue, 0 }  ][line width=0.08]  [draw opacity=0] (8.93,-4.29) -- (0,0) -- (8.93,4.29) -- cycle    ;
\draw    (103,128.77) .. controls (142.2,99.37) and (171.8,167.96) .. (210.61,142.48) ;
\draw [shift={(213,140.8)}, rotate = 143.13] [fill={rgb, 255:red, 0; green, 0; blue, 0 }  ][line width=0.08]  [draw opacity=0] (8.93,-4.29) -- (0,0) -- (8.93,4.29) -- cycle    ;
\draw    (135,195.77) .. controls (165.69,144.09) and (165.02,212.37) .. (213.52,162.35) ;
\draw [shift={(215,160.8)}, rotate = 133.35] [fill={rgb, 255:red, 0; green, 0; blue, 0 }  ][line width=0.08]  [draw opacity=0] (8.93,-4.29) -- (0,0) -- (8.93,4.29) -- cycle    ;
\draw    (463,134.77) .. controls (502.2,105.37) and (487.62,170.09) .. (524.67,144.45) ;
\draw [shift={(527,142.77)}, rotate = 143.13] [fill={rgb, 255:red, 0; green, 0; blue, 0 }  ][line width=0.08]  [draw opacity=0] (8.93,-4.29) -- (0,0) -- (8.93,4.29) -- cycle    ;
\draw  [dash pattern={on 4.5pt off 4.5pt}] (39,43.77) -- (131,43.77) -- (131,151.77) -- (39,151.77) -- cycle ;
\draw    (590,137.77) .. controls (597.92,236.77) and (626.42,125.05) .. (639.61,203.34) ;
\draw [shift={(640,205.77)}, rotate = 261.1] [fill={rgb, 255:red, 0; green, 0; blue, 0 }  ][line width=0.08]  [draw opacity=0] (8.93,-4.29) -- (0,0) -- (8.93,4.29) -- cycle    ;
\draw    (740,137.77) .. controls (721.19,202.12) and (632.79,130.23) .. (639.76,203.5) ;
\draw [shift={(640,205.77)}, rotate = 263.33] [fill={rgb, 255:red, 0; green, 0; blue, 0 }  ][line width=0.08]  [draw opacity=0] (8.93,-4.29) -- (0,0) -- (8.93,4.29) -- cycle    ;
\draw   (391,119) -- (461,119) -- (461,159) -- (391,159) -- cycle(455,125) -- (397,125) -- (397,153) -- (455,153) -- cycle ;
\draw    (321,139.8) .. controls (360.2,110.4) and (349.46,169.32) .. (386.66,143.46) ;
\draw [shift={(389,141.77)}, rotate = 143.13] [fill={rgb, 255:red, 0; green, 0; blue, 0 }  ][line width=0.08]  [draw opacity=0] (8.93,-4.29) -- (0,0) -- (8.93,4.29) -- cycle    ;
\draw  [fill={rgb, 255:red, 0; green, 0; blue, 0 }  ,fill opacity=1 ] (401,132) .. controls (401,129.79) and (402.79,128) .. (405,128) .. controls (407.21,128) and (409,129.79) .. (409,132) .. controls (409,134.21) and (407.21,136) .. (405,136) .. controls (402.79,136) and (401,134.21) .. (401,132) -- cycle ;
\draw  [fill={rgb, 255:red, 0; green, 0; blue, 0 }  ,fill opacity=1 ] (427,132) .. controls (427,129.79) and (428.79,128) .. (431,128) .. controls (433.21,128) and (435,129.79) .. (435,132) .. controls (435,134.21) and (433.21,136) .. (431,136) .. controls (428.79,136) and (427,134.21) .. (427,132) -- cycle ;
\draw  [fill={rgb, 255:red, 0; green, 0; blue, 0 }  ,fill opacity=1 ] (412,141) .. controls (412,138.79) and (413.79,137) .. (416,137) .. controls (418.21,137) and (420,138.79) .. (420,141) .. controls (420,143.21) and (418.21,145) .. (416,145) .. controls (413.79,145) and (412,143.21) .. (412,141) -- cycle ;
\draw  [fill={rgb, 255:red, 0; green, 0; blue, 0 }  ,fill opacity=1 ] (442,143) .. controls (442,140.79) and (443.79,139) .. (446,139) .. controls (448.21,139) and (450,140.79) .. (450,143) .. controls (450,145.21) and (448.21,147) .. (446,147) .. controls (443.79,147) and (442,145.21) .. (442,143) -- cycle ;
\draw    (405,132) -- (416,141) ;
\draw    (446,143) -- (416,141) ;
\draw    (416,141) -- (431,132) ;
\draw    (446,143) -- (431,132) ;
\draw    (555,140.77) .. controls (557.97,194.23) and (686.39,119.29) .. (710.31,207.06) ;
\draw [shift={(711,209.77)}, rotate = 256.55] [fill={rgb, 255:red, 0; green, 0; blue, 0 }  ][line width=0.08]  [draw opacity=0] (8.93,-4.29) -- (0,0) -- (8.93,4.29) -- cycle    ;
\draw    (703,138.77) .. controls (739.08,184.59) and (699.1,175.28) .. (710.05,207.23) ;
\draw [shift={(711,209.77)}, rotate = 248.2] [fill={rgb, 255:red, 0; green, 0; blue, 0 }  ][line width=0.08]  [draw opacity=0] (8.93,-4.29) -- (0,0) -- (8.93,4.29) -- cycle    ;
\draw    (713,230.77) .. controls (739.6,267.21) and (680.81,240.59) .. (714.4,279.92) ;
\draw [shift={(716,281.77)}, rotate = 228.62] [fill={rgb, 255:red, 0; green, 0; blue, 0 }  ][line width=0.08]  [draw opacity=0] (8.93,-4.29) -- (0,0) -- (8.93,4.29) -- cycle    ;
\draw    (714,303.77) .. controls (740.6,340.21) and (681.81,313.59) .. (715.4,352.92) ;
\draw [shift={(717,354.77)}, rotate = 228.62] [fill={rgb, 255:red, 0; green, 0; blue, 0 }  ][line width=0.08]  [draw opacity=0] (8.93,-4.29) -- (0,0) -- (8.93,4.29) -- cycle    ;
\draw  [dash pattern={on 4.5pt off 4.5pt}] (616,189.77) -- (788,189.77) -- (788,393.77) -- (616,393.77) -- cycle ;
\draw    (702,394.77) .. controls (728.6,431.21) and (669.81,404.59) .. (703.4,443.92) ;
\draw [shift={(705,445.77)}, rotate = 228.62] [fill={rgb, 255:red, 0; green, 0; blue, 0 }  ][line width=0.08]  [draw opacity=0] (8.93,-4.29) -- (0,0) -- (8.93,4.29) -- cycle    ;
\draw  [dash pattern={on 4.5pt off 4.5pt}] (493,89.77) -- (811,89.77) -- (811,235.77) -- (493,235.77) -- cycle ;
\draw   (686,44) .. controls (686,32.95) and (701.67,24) .. (721,24) .. controls (740.33,24) and (756,32.95) .. (756,44) .. controls (756,55.05) and (740.33,64) .. (721,64) .. controls (701.67,64) and (686,55.05) .. (686,44) -- cycle ;
\draw    (756,44) .. controls (773.55,45.72) and (774.94,54.32) .. (769.44,89.04) ;
\draw [shift={(769,91.77)}, rotate = 279.21] [fill={rgb, 255:red, 0; green, 0; blue, 0 }  ][line width=0.08]  [draw opacity=0] (8.93,-4.29) -- (0,0) -- (8.93,4.29) -- cycle    ;
\draw    (682.48,43.9) .. controls (646.38,43.9) and (668.6,78.1) .. (688,88.77) ;
\draw [shift={(686,44)}, rotate = 183.12] [fill={rgb, 255:red, 0; green, 0; blue, 0 }  ][line width=0.08]  [draw opacity=0] (8.93,-4.29) -- (0,0) -- (8.93,4.29) -- cycle    ;
\draw   (813,296) .. controls (813,284.95) and (831.8,276) .. (855,276) .. controls (878.2,276) and (897,284.95) .. (897,296) .. controls (897,307.05) and (878.2,316) .. (855,316) .. controls (831.8,316) and (813,307.05) .. (813,296) -- cycle ;
\draw    (789,249) .. controls (820.36,242.89) and (857.48,232.2) .. (855.19,273.41) ;
\draw [shift={(855,276)}, rotate = 275.17] [fill={rgb, 255:red, 0; green, 0; blue, 0 }  ][line width=0.08]  [draw opacity=0] (8.93,-4.29) -- (0,0) -- (8.93,4.29) -- cycle    ;
\draw    (855,316) .. controls (858.88,346.81) and (823.24,342.09) .. (791.89,343.61) ;
\draw [shift={(789,343.77)}, rotate = 356.42] [fill={rgb, 255:red, 0; green, 0; blue, 0 }  ][line width=0.08]  [draw opacity=0] (8.93,-4.29) -- (0,0) -- (8.93,4.29) -- cycle    ;

\draw (41,58) node [anchor=north west][inner sep=0.75pt]   [align=left] {LTL property};
\draw (78,75) node [anchor=north west][inner sep=0.75pt]   [align=left] {$\displaystyle \varphi $};
\draw (62,113) node [anchor=north west][inner sep=0.75pt]   [align=left] {iCGS};
\draw (72,130) node [anchor=north west][inner sep=0.75pt]   [align=left] {$\displaystyle M$};
\draw (88,198) node [anchor=north west][inner sep=0.75pt]   [align=center] {Execution \\trace};
\draw (116,235) node [anchor=north west][inner sep=0.75pt]   [align=left] {$\displaystyle h$};
\draw (269,138.9) node  [font=\large] [align=left] {Our tool};
\draw (40,22) node [anchor=north west][inner sep=0.75pt]   [align=left] {.json};
\draw (515,116) node [anchor=north west][inner sep=0.75pt]   [align=left] {$\displaystyle \{\langle s_{i} ,\psi _{i} ,atom_{\psi _{i}} \rangle ,\dotsc ,\ \langle s_{j} ,\psi _{j} ,atom_{\psi _{j}} \rangle \}$};
\draw (632,207) node [anchor=north west][inner sep=0.75pt]   [align=left] {$\displaystyle M'$};
\draw (427.09,181.5) node   [align=center] {internal \\representation};
\draw (682,208) node [anchor=north west][inner sep=0.75pt]   [align=left] {$\displaystyle \langle \psi _{n} ,\psi _{p} \rangle $};
\draw (684,282) node [anchor=north west][inner sep=0.75pt]   [align=left] {$\displaystyle \langle \varphi _{n} ,\varphi _{p} \rangle $};
\draw (723,246) node [anchor=north west][inner sep=0.75pt]   [align=left] {to LTL};
\draw (664,358) node [anchor=north west][inner sep=0.75pt]   [align=left] {$\displaystyle \langle Mon_{\varphi _{n}} ,Mon_{\varphi _{p}} \rangle $};
\draw (331,106) node [anchor=north west][inner sep=0.75pt]   [align=left] {parsing};
\draw (766,398) node [anchor=north west][inner sep=0.75pt]   [align=left] {RV};
\draw (642,447) node [anchor=north west][inner sep=0.75pt]   [align=left] {$\displaystyle \langle k,\varphi _{mc} ,\varphi _{rv} ,\varphi _{unchk} \rangle $};
\draw (789,69) node [anchor=north west][inner sep=0.75pt]   [align=left] {MC};
\draw (721,44) node   [align=left] {MCMAS};
\draw (855,296) node   [align=left] {LamaConv};

\end{tikzpicture}

}
\caption{Overview of the implemented tool.}
\label{fig:tool}
\end{figure}
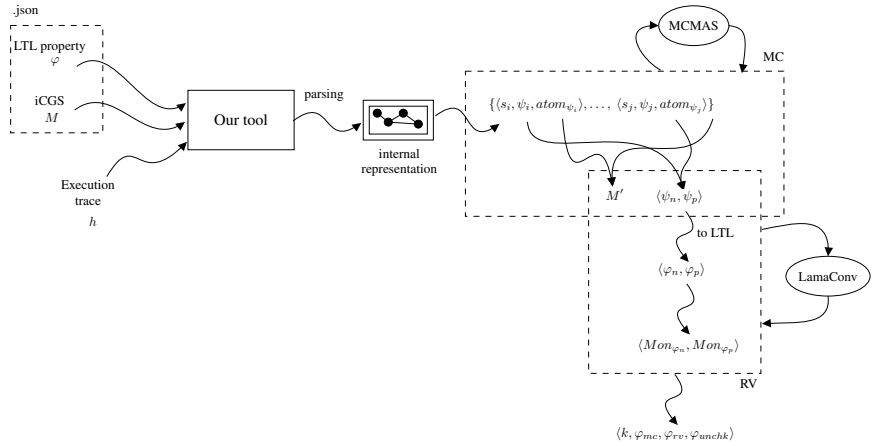

The entire manipulation, from parsing the model formatted in Json, to translating the latter to its equivalent ISPL program, has been performed by extending an existent Java library~\cite{StrategicTool}; the rest of the tool derives directly from the algorithms presented in this paper. 
The monitors generated by Algorithm~\ref{alg:rv} at lines 18 and 19 are obtained using LamaConv~\cite{LamaConv}, which is a Java library capable of translating expressions in temporal logic into equivalent automata and generating monitors out of these automata. For generating monitors, LamaConv uses the algorithm presented in~\cite{DBLP:journals/tosem/BauerLS11}.  

\vspace{-0.5em}
\subsection{Experiments} 

We tested our tool on a large set of automatically and randomly generated iCGSs; on a machine with the following specifications: Intel(R) Core(TM) i7-7700HQ CPU @ 2.80GHz, 4 cores 8 threads, 16 GB RAM DDR4. 
The objective of these experiments was to show how many times our algorithm returned a conclusive verdict. For each model, we ran our procedure and counted the number of times a solution was returned. Note that, our approach concludes in any case, but since the general problem is undecidable, the result might be inconclusive (\textit{i.e.}, $?$). In Figure~\ref{fig:experiments-success}, we report our results by varying the percentage of imperfect information (x axis) inside the iCGSs, from $0\%$ (perfect information, \textit{i.e.}, all states are distinguishable for all agents), to $100\%$ (no information, \textit{i.e.}, no state is distinguishable for any agent). For each percentage selected, we generated $10000$ random iCGSs and counted the number of times our algorithm returned with a conclusive result (\textit{i.e.}, $\top$ or $\bot$). As it can be seen in Figure~\ref{fig:experiments-success}, our tool concludes with a conclusive result more than 80\% of times. We do not observe any relevant difference amongst the different percentage of information used in the experiments. This is mainly due to the completely random nature of the iCGSs used. In more detail, the results we obtained completely depend on the topology of the iCGSs, so it is very hard to precisely quantify the success rate. However, the results obtained by our experiments using our procedure are encouraging. Unfortunately, no benchmark of existing iCGSs -- to test our tool on -- exists, thus these results may vary on more realistic scenarios. Nonetheless, considering the large set of iCGSs we experimented on, we do not expect substantial differences.

\begin{figure}
    \centering
    \includegraphics[width=\linewidth]{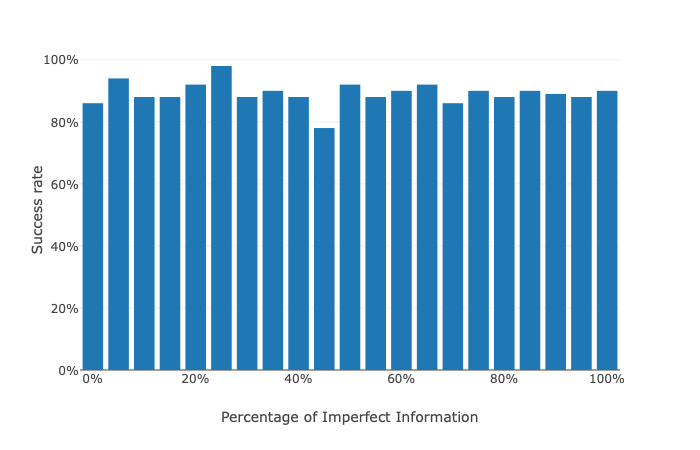}
    \caption{Success rate of our tool when applied to a set of randomly generated iCGSs.}
    \label{fig:experiments-success}
\end{figure}

Other than testing our tool w.r.t. the success rate over a random set of iCGSs, we evaluated the execution time as well. Specifically, we were much interested in analysing how such execution time is divided between $CheckSub$-$formulas()$ and Algorithm~\ref{alg:rv}. \textit{I.e.}, how much time is spent on verifying the models statically (through model checking), and how much is spent on verifying the temporal properties (through runtime verification). Figure~\ref{fig:experiments-time} reports the results we obtained on the same set of randomly generated used in Figure~\ref{fig:experiments-success}. The results we obtained are intriguing, indeed we can note a variation in the percentage of time spent on the two phases (y-axis) moving from low percentages to high percentages of imperfect information in the iCGSs (x-axis). When the iCGS is close to have perfect information (low percentages on x-axis), we may observe that most of the execution time is spent on performing static verification ($\sim$70\%), which corresponds to $CheckSub$-$formulas()$. On the other hand, when imperfect information grows inside the iCGS (high percentage on x-axis), we may observe that most of the execution time is spent on performing runtime verification ($\sim$90\% in occurrence of absence of information). The reason for this change in the execution behaviour is determined by the number of candidates extracted by the $FindSub$-$models()$ function. When the iCGS has perfect information, such function only extracts a single candidate (\textit{i.e.}, the entire model), since $FindSub$-$models()$ generates only one tuple. Such single candidate can be of non-negligible size, and the resulting static verification, time consuming; while the subsequent runtime verification is only performed once on the remaining temporal parts of the property to verify. On the other hand, when the iCGS has imperfect information, $FindSub$-$models()$ returns a set of candidates that can grow exponentially w.r.t. the number of states of the iCGS. Nonetheless, such candidates are small in size, since $FindSub$-$models()$ splits the iCGS into multiple smaller iCGSs with perfect information. Because of this, the static verification step is applied on small iCGSs and require less execution time; while the runtime verification step is called for each candidate (so an exponential number of times) and is only influenced by the size of the temporal property to verify.

\begin{figure}
    \centering
    \includegraphics[width=\linewidth]{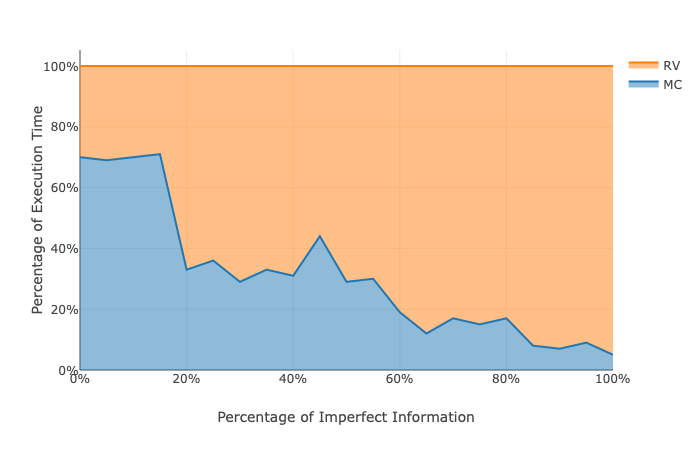}
    \caption{How the execution time of our tool when applied to a set of randomly generated iCGSs is divided.}
    \label{fig:experiments-time}
\end{figure}


In conclusion, it is important to emphasise that, even though the monitor synthesis is computationally hard (\textit{i.e.}, $2EXPTIME$), the resulting runtime verification process is polynomial in the size of the history analysed.
Naturally, the actual running complexity of a monitor depends on the formalism used to describe the formal property. In this work, monitors are synthesised from LTL properties. Since LTL properties are translated into Moore machines~\cite{DBLP:journals/tosem/BauerLS11}; because of this, the time complexity w.r.t. the length of the analysed trace is linear. This can be understood intuitively by noticing that the Moore machine so generated has finite size, and it does not change at runtime. Thus, the number of execution steps for each event in the trace is constant.


\section{Conclusions and Future work} \label{sec:conclusions}

The work presented in this paper follows a standard combined approach of formal verification techniques, where the objective is to get the best of both. 
We considered the model checking problem of MAS using strategic properties that is undecidable in general, and showed how runtime verification can help by verifying part of the properties at execution time. The resulting procedure has been presented both on a theoretical (theorems and algorithms) and a practical level (prototype implementation). 
It is important to note that this is the first attempt of combining model checking and runtime verification to verify strategic properties on a MAS. Thus, even though our solution might not be optimal, it is a milestone for the corresponding lines of research. Additional works will be done to improve the technique and, above all, its implementation. For instance, we are planning to extend this work considering a more predictive flavour. This can be done by recognising the fact that by verifying at static time part of the system, we can use this information at runtime to predict future events and conclude the runtime verification in advance.


%



\bibliographystyle{unsrt}
\bibliography{new_bib}


\end{document}